\documentclass[a4paper,11pt]{article}

\usepackage[top=1 in,bottom=1 in,left=1.0 in,right=1.0 in]{geometry}
\usepackage{graphics,amsmath,amssymb,amsthm,accents}
\usepackage{appendix}
\usepackage{epic}
\usepackage{amsfonts}
\usepackage{xcolor}
\usepackage{eucal}
\usepackage{latexsym}

\numberwithin{equation}{section}

\usepackage{array} 
\usepackage{float}
\usepackage{todonotes}
\usepackage{epic,epsfig}
\usepackage{graphicx}
\usepackage[all]{xy}
\usepackage{tikz}
\newtheorem{theorem}{Theorem}[section]

\newtheorem{remark}{Remark}
\numberwithin{equation}{section}
\DeclareMathAccent{\wtilde}{\mathord}{largesymbols}{"65}
\DeclareMathAccent{\what}{\mathord}{largesymbols}{"62}

\def\m@th{\mathsurround=0pt}
\mathchardef\bracell="0365
\def\upbrall{$\m@th\bracell$}
\def\undertilde#1{\mathop{\vtop{\ialign{##\crcr
    $\hfil\displaystyle{#1}\hfil$\crcr
     \noalign
     {\kern1.5pt\nointerlineskip}
     \upbrall\crcr\noalign{\kern1pt
   }}}}\limits}

\newcommand{\bblu}{\begin{color}{blue}}
\newcommand{\bred}{\begin{color}{red}}
\newcommand{\ecl}{\end{color}}

\newcommand{\bB}{\boldsymbol{B}}

\newcommand{\bI}{\boldsymbol{I}}

\newcommand{\bK}{\boldsymbol{K}}

\newcommand{\bM}{\boldsymbol{M}}

\newcommand{\bt}{{\boldsymbol{t}}}
\newcommand{\bi}{{\boldsymbol{i}}}

\newcommand{\be}{\begin{equation}}
\newcommand{\ee}{\end{equation}}
\newcommand{\bea}{\begin{eqnarray}}
\newcommand{\eea}{\end{eqnarray}}
\newcommand{\bse}{\begin{subequations}}
\newcommand{\ese}{\end{subequations}}
\newcommand{\nn}{\nonumber}

\newcommand{\bu}{\boldsymbol{u}}
\newcommand{\bv}{{\boldsymbol v}}

\newcommand{\bs}{{\boldsymbol s}}

\begin{document}
\title{Lattice KP type  equations arising from eigenfunctions \\ and Dbar problems}

\author{Leilei Shi$^{1,2,3,5}$, ~~Peter van der Kamp$^{4}$, ~~ Cheng Zhang$^{1,2}$,~~
Da-jun Zhang$^{1,2}$\footnote{Corresponding author. Email: djzhang@shu.edu.cn}
\\
 [7pt]
{\small $^{1}$Department of Mathematics, Shanghai University, Shanghai 200444,    China} \\
{\small $^{2}$Newtouch Center for Mathematics of Shanghai University,  Shanghai 200444, China}\\
{\small $^{3}$SISSA, International School for Advanced Studies, Trieste 34136, Italy}\\
{\small $^{4}$Department of Mathematical and Physical Sciences, La Trobe University, Victoria 3086, Australia}\\
{\small $^{5}$INFN, Sezione di Trieste, Trieste, Italy}
	}

\maketitle

\begin{abstract}
In this paper, we construct the lattice Kadomtsev-Petviashvili (KP) type eigenfunction equations.
A homogeneous nonlocal $\bar{\partial}$ problem is considered,
from which we are able to define the eigenfunction of the Lax pair of the lattice KP equation.
The eigenfunction together with its expansions at infinity and at a finite analytic point
provide formulations of the lattice modified KP equation, the lattice Schwarzian KP equation
and the Nijhoff-Quispel-Capel KP (NQC-KP) equation.
We also consider an inhomogeneous nonlocal $\bar{\partial}$ problem.
It defines the eigenfunction of the Lax pair of the lattice modified KP equation.
This eigenfunction generates a direct formulation for the NQC-KP equation,
which is different from the previous ones.
Explicit solutions of these equations are obtained, from which we can see the difference
of the different formulations for  same equations.

\begin{description}
\item[Keywords:] lattice Kadomtsev-Petviashvili type equation; eigenfunction; Dbar problem; soliton solution
\item[Mathematics Subject Classification:] 37K60
\end{description}
\end{abstract}

\section{Introduction}\label{sec-1}

Most of integrable systems have explicit Lax pairs.
As solutions of the Lax pairs, the eigenfunctions can be formulated by tau functions,
which has been observed from the bilinear method \cite{Hir-1974},
Sato's theory \cite{MiwJD-2000}, direct linearization (DL) approach \cite{FA-PRL-1981}, etc.
Therefore it is not surprised that many integrable equations can be characterized by eigenfunctions.
A well known example is the modified Korteweg-de Vries (mKdV) equation
by the eigenfunction of the Lax pair of the KdV equation
(in this paper, we call it the eigenfunction of the KdV equation for short) via the Cole-Hopf transformation.
Konopelchenko used to study eigenfunction equations of some continuous integrable equations \cite{K-RMP-1990}.
In a recent paper \cite{WZZZ-JPA-2024}, eigenfunction equations related to the lattice KdV (lKdV)
equation were investigated.
For lattice eigenfunction equations were constructed, from which some equations in the
Adler-Bobenko-Suris (ABS) list \cite{ABS-CMP-2003} (except H1, H2 and Q4)
can be formulated by the eigenfunctions of the lKdV equation \cite{WZZZ-JPA-2024}.
More recently, we revisited the construction of the lKdV-type eigenfunction equations \cite{SZZ-MPAG-2025,SZZ-PA-2026},
facilitated by combining the  $\bar{\partial}$ (Dbar) method and DL approach.

The $\bar{\partial}$  problem
\begin{eqnarray}\label{1.1}
    \bar{\partial}f(p,\bar{p}):=\frac{\partial}{\partial \bar{p}}f(p,\bar{p})=g(p,\bar{p})
\end{eqnarray}
is considered as a fundamental equation in complex analysis for solving the unknown function $f$
for a given function $g$.
Here $p\in\mathbb{C}$ and $\bar{p}$ stands for its complex conjugate.
Since $\bar{\partial}f|_{p=p_0}=0$ indicates $f$ is analytic at $p_0\in\mathbb{C}$,
the function $\bar{\partial}f$ serves a role to measure  the deviation of $f$ from holomorphicity.
The $\bar{\partial}$  problem as a concept was introduced by Beals and Coifman \cite{BC-1980-81,BC-1981-82}.
It was used on a high order matrix spectral problem, constructing solutions via  the inverse scattering transform (IST) \cite{BC-1981-82,BC-CAPM-1984}.
As a generalization of the Riemann-Hilbert problem,
it is in particular crucial in using the IST to solve
those (2+1)-dimensional equations (e.g. the Kadomtsev-Petviashvili (KP) II equation \cite{ABF-SAPM-1983})
of which the Riemann-Hilbert formulation is inadequate.
Combined with the dressing method, the $\bar{\partial}$-dressing approach  has been systematically
described by Zakharov and Manakov in \cite{ZM-FAA-1985}
and Beals and Coifman in \cite{BC-PD-1986} (also see the review paper \cite{BC-IP-1989}).
In addition, the $\bar{\partial}$ problem can be embedded into the DL scheme
through the generalized Cauchy integral formula (also known as the Cauchy-Pompeiu formula
where the function $g$ in \eqref{1.1} should be considered as a distribution),
which was  observed by Fokas and Ablowitz in \cite{FA-JMP-1984}.
Meanwhile, Jaulent and Manna developed a direct way, other than the $\bar{\partial}$-dressing,
in a sequence of works \cite{JM-IP-1986,JM-EPL-1986,JM-IP-1987,JM-JMP-1987}.
They considered  solutions of the $\bar{\partial}$ problem as eigenfunctions,
formulated by the associated  Cauchy-Pompeiu formula.
Using these solutions, they introduced the so-called spectral Wronskian (see \cite{JM-JMP-1987}),
by which they can define a squared eigenfunction symmetry
as well as introduce and deal with time evolution for the whole integrable hierarchy.
Their method was later illustrated with details for the KdV hierarchy and Toda hierarchy \cite{JMM-IP-1988}.

In our recent paper \cite{SZZ-MPAG-2025}, we combined a $\bar{\partial}$ problem and the DL approach,
by which we could recover all the lattice eigenfunctions found in \cite{WZZZ-JPA-2024}.
A quasi-local $\bar{\partial}$ problem was introduced and  an eigenfunction was defined
through the associated  Cauchy-Pompeiu formula.
The eigenfunction is proved to satisfy the Lax pair of the lattice potential KdV (lpKdV) equation.
And the eigenfunction itself satisfies the lattice mKdV (lmKdV) equation
(or H3 equation in the ABS list \cite{ABS-CMP-2003}).
Later, in \cite{SZZ-PA-2026}, we started from the Lax pair of  the lpmKdV equation (with a parameter extension
in the Lax pair)
and introduced an inhomogeneous quasi-local $\bar{\partial}$ problem.
We could show that the eigenfunction of the lmKdV equation gives rise to the
Nijhoff-Quispel-Capel (NQC) equation.
Thus, by means of $\bar{\partial}$ problems and the DL approach, we could see how the lKdV-type equations
arise as eigenfunction equations from the Lax pairs of the lpKdV equation and the lmKdV equation.
In addition, in \cite{SZZ-MPAG-2025,SZZ-PA-2026}, we also explained how the Cauchy matrix
variables $\{S^{(i,j)}\}$ are formulated from these eigenfunctions.

To have a full profile for the eigenfunction equations of lattice equations,
we extend the research of \cite{SZZ-MPAG-2025,SZZ-PA-2026} to the 3-dimensional case.
We will consider two nonlocal $\bar{\partial}$ problems.
The first one will be connected to the eigenfunction of the lattice KP (lKP) equation. We will see that
the  eigenfunction directly defines the lattice mKP (lmKP) equation.
We will also show that how the lattice Schwarzian KP (lSKP) equation arise from this approach.
Besides, we will show two different formulations to the NQC-KP equation from the lKP eigenfunction.
The second one is an inhomogeneous nonlocal $\bar{\partial}$ problem,
connected with the Lax pair of the  lmKP equation.
The eigenfunction will provide  a third way to formulate  the NQC-KP equation.
In addition to the lKP-type equations, we will also
discuss the correspondence between the  lKP-type  Cauchy matrix
variables  and the lKP and lmKP eigenfunctions.

The paper is organized as follows.
In Sec.\ref{sec-2}, we work on a homogeneous nonlocal $\bar{\partial}$ problem
that is associated with the lKP equation and its Lax pair.
The lmKP, lSKP and NQC-KP equations will be constructed from the lKP eigenfunction.
$N$-soliton solutions of these equations together with
some lKP-type  Cauchy matrix variables  will be formulated.
In Sec.\ref{sec-3}, we introduce an inhomogeneous $\bar{\partial}$ problem
to connect the lmKP equation and show that the lmKP eigenfunction
directly defines the NQC-KP equation.
Their $N$-soliton solutions will be presented as well.
Finally, conclusions are provided in Sec.\ref{sec-4}.

\section{Homogeneous nonlocal $\bar{\partial}$ problem and lKP-type equations}\label{sec-2}

In this section, we will work on a homogeneous nonlocal $\bar{\partial}$ problem (see\eqref{DP-1}),
by which we define the eigenfunction of the Lax pair of the lKP equation.
The lmKP, lSKP and NQC-KP equations together with their soliton solutions will be formulated
from this eigenfunction.

\subsection{Homogeneous nonlocal $\bar{\partial}$ problem and the lKP equation}\label{sec-2-1}

The lattice potential KP (lpKP) equation is \cite{NCWQ-PLA-1984}
\begin{eqnarray}\label{dKP}
    (\dot{u}+\zeta)(\alpha-\beta+\dot{\widetilde{u}}
    -\dot{\widehat{u}})+(\widetilde{u}+\alpha)(\beta-\zeta+\widehat{\widetilde{u}}-\dot{\widetilde{u}})
    +(\widehat{u}+\beta)(\zeta-\alpha+\dot{\widehat{u}}-\widehat{\widetilde{u}})=0,
\end{eqnarray}
where $u:=u(n,m,l)$  is a function defined on $\mathbb{Z}^3$,
$\alpha, \beta$ and $\zeta$ are the spacing parameters in the $n, m$ and $l$-direction, respectively.
In this paper, we employ shorthand notations to express shifts of a function in different directions, e.g.
\begin{equation}
\widetilde{u}:=u(n+1,m,l),~~\widehat{u}:=u(n,m+1,l),~~ \dot{u}:=u(n,m,l+1),~~
\widehat{\widetilde{u}}:=u(n+1,m+1,l).
\end{equation}

Let us consider the homogeneous nonlocal $\bar{\partial}$ problem
\begin{subequations}\label{DP-1}
\begin{eqnarray}\label{DP-11}
\bar{\partial} \phi\left( p\right)
=\int_\mathbb{C}\phi \left( \mu\right) R\left(\mu, p\right){\rm d}\mu\wedge {\rm d}\bar{\mu},
\end{eqnarray}
with the asymptotic behavior
\begin{equation}\label{chicond}
    \phi \left( p\right) =1+p^{-1}\phi^{(-1)}+p^{-2}\phi^{(-2)}+\dots, \ \ |p|\rightarrow \infty.
\end{equation}
Here and hereafter, we adopt the shorthand $\phi(p)$ for $\phi(n,m,l; p,\bar{p})$.
The kernel $R(\mu,p)$ is chosen to be of the following separable form
\begin{eqnarray}\label{R}
    R(\mu,p)=(r_1(\mu))^T r_2(p)=\rho(\mu)(r^0_1(\mu))^T r^0_2(p) \rho^{-1}(p),
\end{eqnarray}
where
\begin{eqnarray}\label{rho}
 \rho(p)=(\alpha+p)^n(\beta+p)^m(\zeta+p)^l,
\end{eqnarray}
$r_1(\mu)=\rho(\mu) r^0_1(\mu)$, $r_2(p)= \rho^{-1}(p)  r^0_2(p)$, $r^0_1(p)$ and $r^0_2(p)$ are $N$-component column vector functions
depending only on $p$ (independent of $n,m$ and $l$).
In addition, we impose the condition on $r_1(p),r_2(p) \in L_{0,0}(\mathbb{C})$ which satisfy the hypotheses of Theorem~\ref{thm-1}. Here we denote by $L_{k,s}(V )$ the space of $(k,s)$-forms on $V$ with coefficients in distributions of measure type on $V$ \cite{H-2012}.
\end{subequations}
Note again that with the above setting, $\phi(p)$ also depends on  $(n,m,l)$  and their parameters $(\alpha,\beta,\zeta)$.

To solve the nonlocal $\bar{\partial}$ problem \eqref{DP-1}, by making use of
the Cauchy-Pompeiu integral formula, is equivalent to solving the following linear integral equation
\begin{equation}\label{LI}
    \phi \left( p\right) =1+\frac{1}{2\pi {\rm i}}\int _\mathbb{C}\int _\mathbb{C}
    \frac{\phi(\mu) R\left( \mu,\xi \right) }{\xi -p}
    {\rm d}\mu \wedge {\rm d}\overline{\mu }\wedge {\rm d}\xi \wedge {\rm d}\overline{\xi }.
\end{equation}
We will explain this equivalence later through
the following remark and Theorem \ref{thm-1}.

\begin{remark}\label{Rem-1}
The Cauchy-Pompeiu integral formula is \cite{P-1909} (see also \cite{H-2012,AF-book-2021})
\begin{eqnarray}\label{Cauchy}
    f(p) = \frac{1}{2\pi {\rm i}} \int_{\partial \Omega} \frac{f(\mu)}{\mu - p} \,
    \mathrm{d}\mu + \frac{1}{2\pi {\rm i}} \int_{\Omega} \frac{\bar{\partial}f(\mu)}{\mu - p} \, \mathrm{d}\mu \wedge \mathrm{d}\bar{\mu},
\end{eqnarray}
where $f(p)$ is continuous and has continuous partial derivatives in a finite region $\Omega$
with the boundary $\partial \Omega$ that is a positively oriented simple closed contour.
This formula remains valid when $\bar{\partial}f(p)$ is a general function or a distribution
and then $\Omega$ can be extended to the whole complex plane $\mathbb{C}$.
(see \cite{H-2012} and  Appendix A in \cite{SZZ-PA-2026}).
In this case, if  $f(p)$ is a solution of the $\bar{\partial}$ problem \eqref{DP-11}
and  has an expansion at infinity:
\begin{equation}
f(p)=\sum^N_{j=-\infty} f^{(j)} p^j, ~~~ (N \geq 0),
\end{equation}
then, substituting the $\bar{\partial}$ problem \eqref{DP-11} into \eqref{Cauchy} gives rise to an integral equation
\begin{eqnarray}\label{CP}
    f(p) = \mathrm{Nor}f(p)
    + \frac{1}{2\pi {\rm i}}\int _\mathbb{C}\int _\mathbb{C}\frac{f(\mu) R\left( \mu,\xi \right) }{\xi -p} {\rm d}\mu \wedge {\rm d}\overline{\mu }\wedge {\rm d}\xi \wedge {\rm d}\overline{\xi }.
\end{eqnarray}
where $\mathrm{Nor}f(p)=\sum^N_{j=0} f^{(j)} p^j$
denotes the principal part of $f(p)$ at infinity  \cite{JMM-IP-1988}.
Thus, if we only consider those solutions that are determined through solving the integral equation \eqref{CP},
the assumption that the homogeneous equation admits only the zero solution $f(p)=0$
is equivalent to the assumption that \eqref{DP-11} admits only zero solution $f(p)=0$ if $\mathrm{Nor}f(p)=0$.
Such an assumption has been used in the paper \cite{JMM-IP-1988}.
\end{remark}

In the following, we will show that the $\bar{\partial}$ problem \eqref{DP-1} does admit a unique solution.

\begin{theorem}\label{thm-1}
Define
\begin{equation}\label{g}
g(p)
:=
\frac{1}{2\pi\mathrm{i}}
\int_{\mathbb C}
\frac{r_2(\xi)}{\xi-p}\,
{\rm d}\xi\wedge {\rm d}\bar\xi
\end{equation}
and
\begin{equation}\label{K}
\mathcal{K}
:
=
-\frac{1}{2\pi\mathrm{i}}
\int_{\mathbb C}\int_{\mathbb C}
\frac{r_2(\xi)r_1^T(\mu)}
{\xi-\mu}\,
{\rm d}\xi\wedge {\rm d}\bar\xi\,
{\rm d}\mu\wedge {\rm d}\bar\mu.
\end{equation}
Suppose that all the above integrals are well defined, and
\[
\det(\bI_N+\mathcal{K})\neq0,
\]
then the $\bar\partial$ problem \eqref{DP-1} admits a unique solution.
\end{theorem}

\begin{proof}
We first prove existence of the solution. 
Note that from its definition $\mathcal{K}$ is independent of both $p$ and $\bar{p}$,
and it can be written as
\begin{equation}
\mathcal{K}=- \int_{\mathbb C}
g(\mu)r_1^T(\mu)\,
{\rm d}\mu\wedge {\rm d}\bar\mu.
\end{equation}
Set
\begin{equation}
a_0
:=
\int_{\mathbb C}
r_1^T(\mu)\,
{\rm d}\mu\wedge {\rm d}\bar\mu .
\end{equation}
Since $I_N+\mathcal{K}$ is invertible, we define
\begin{equation}\label{A-0}
A_0:=a_0(\bI_N+\mathcal{K})^{-1}
\end{equation}
and
\begin{eqnarray}\label{phi-construction}
    \phi_0(p):=1+A_0g(p).
\end{eqnarray}
We will show that $\phi_0(p)$ is a solution to the Dbar problem \eqref{DP-1}.
Noting that the Green function of the $\bar \partial$ is $\frac{-1}{2\pi \mathrm{i} (\xi-p)}$,
from \eqref{g} we immediately have $\bar \partial g(p)=r_2(p)$,
and consequently,
\[
\bar\partial \phi_0(p)=A_0r_2(p),
\]
which is subject to
\[
\phi_0(p)=1+O(p^{-1}).
\]
Now, multiplying both sides of \eqref{phi-construction} by $r_1^T(p)$ and integrating over $\mathbb C$, we obtain
\begin{align*}
\int_{\mathbb C}
\phi_0(\mu)r_1^T(\mu)\,
{\rm d}\mu\wedge {\rm d}\bar\mu
&=
a_0+
A_0\int_{\mathbb C}
g(\mu)r_1^T(\mu)\,
{\rm d}\mu\wedge {\rm d}\bar\mu
\\
&=
a_0-A_0\mathcal{K}
\\
&=
A_0,
\end{align*}
where the last two equalities follow from the \eqref{K} and \eqref{A-0}.
Finally, noticing that $R(\mu,p)$ is separated as in \eqref{R}, we have
\[
\int_{\mathbb C}
\phi_0(\mu)R(\mu,p)\,
{\rm d}\mu\wedge {\rm d}\bar\mu
=
A_0 r_2(p)
=
\bar\partial \phi_0(p),
\]
which means $\phi_0$ is a solution of the $\bar\partial$ problem \eqref{DP-1}.

We next prove uniqueness. Let $\phi$ be any solution of the $\bar\partial$ problem \eqref{DP-1}. 
Since $ R(\mu,p)$ is separated as in \eqref{R}, we can rewrite \eqref{DP-11} into 
\[
\bar\partial \phi(p)=Ar_2(p),
\]
where
\[
A
:=
\int_{\mathbb C}
\phi(\mu)r_1^T(\mu)\,
{\rm d}\mu\wedge {\rm d}\bar\mu.
\]
Then, making use of the generalized Cauchy-Pompeiu integral formula \eqref{CP}, we have
\[
\phi(p)=1+Ag(p).
\]
Substituting this expression into the definition of $A$ yields
\begin{align*}
A
&=
a_0+
A\int_{\mathbb C}
g(\mu)r_1^T(\mu)\,
{\rm d}\mu\wedge {\rm d}\bar\mu.
\\
&=
a_0-A\mathcal{K}.
\end{align*}
Since $\bI_N+\mathcal{K}$ is invertible, we have
\[
A=a_0(\bI_N+\mathcal{K})^{-1}=A_0,
\]
and hence
\[
\phi(p)=1+Ag(p)=1+A_0g(p)=\phi_0(p).
\]
Thus the solution of the $\bar\partial$ problem \eqref{DP-1} exists and is unique.
\end{proof}

\subsection{Lax pair, lKP and lpKP equations}\label{sec-2-2}

We now show that the eigenfunction $\phi(p)$, defined by the $\bar{\partial}$ problem \eqref{DP-1}, satisfies
a Lax pair  which will later lead us to the lKP and lpKP equations.
First, we define an operator $L(p)$ by acting on $\phi(p)$:
\begin{subequations}
    \begin{align}
       L(p)\phi(p) =(\beta+p)\widehat{\phi}(p)-(\alpha+p)\widetilde{\phi}(p)+h\phi(p),
   \end{align}
\end{subequations}
where
\begin{subequations}
    \begin{align}\label{h}
           h =\alpha-\beta+\widetilde{\phi}^{(-1)}-\widehat{\phi}^{(-1)},
    \end{align}
\end{subequations}
Then we consider the asymptotic behavior of $L(p)\phi(p)$ in light of \eqref{chicond}:
\begin{align}
     L(p)\phi(p)&=(\beta+p)(1+p^{-1}\widehat{\phi}^{(-1)}+\cdots)
     -(\alpha+p)(1+p^{-1}\widetilde{\phi}^{(-1)}+\cdots)+h(1+p^{-1}\phi^{(-1)}+\cdots)\nonumber\\
     &=\beta-\alpha+\widehat{\phi}^{(-1)}-\widetilde{\phi}^{(-1)}+h+O(p^{(-1)})\nonumber\\
     &=O(p^{(-1)})\nonumber,
\end{align}
and compute
\begin{align}
    \bar{\partial}\big(L(p)\phi(p)\big)
    &=(\beta+p)\bar{\partial}\widehat{\phi}(p)-(\alpha+p)\bar{\partial}\widetilde{\phi}(p)
    +h\bar{\partial}\phi(p)\nonumber\\
    &=\int_C(\beta+\mu)\widehat{\phi} \left( \mu\right)R\left(\mu, p\right)-(\alpha+\mu)\widetilde{\phi}
    \left( \mu\right)R\left(\mu, p\right)+h\phi(\mu)R(\mu,p){\rm d}\mu{\rm d}\bar{\mu}\nonumber\\
    &=\int_CL(\mu)\phi(\mu)R(\mu,p){\rm d}\mu{\rm d}\bar{\mu}\nonumber,
\end{align}
where we note that $h$ is independent of \(p\).
Thus $L(p)\phi(p)$ satisfies the same $\bar{\partial}$ problem \eqref{DP-11} as $\phi(p)$ and  $\mathrm{Nor}\big(L(p)\phi(p)\big)=0$.
Therefore, due to the uniqueness of the solutions for the  $\bar{\partial}$ problem \eqref{DP-1}
stated in Theorem \ref{thm-1},
we get the first linear equation for $\phi(p)$:
\begin{subequations}\label{Lax_tri}
\begin{align}
   L(p)\phi(p)=   (\beta+p)\widehat{\phi}(p)-(\alpha+p)\widetilde{\phi}(p)+h\phi(p)&=0\label{L1}.
\end{align}
Along the same line one can verify that
\begin{align}
     (\zeta+p)\dot{\phi}(p)-(\alpha+p)\widetilde{\phi}(p)+f\phi(p) =0,\label{L2}
    \end{align}
\end{subequations}
where
\begin{eqnarray}\label{f}
        f =\alpha-\zeta+\widetilde{\phi}^{(-1)}-\dot{\phi}^{(-1)}.
\end{eqnarray}
Note that \eqref{L1} and \eqref{L2} indicate that
\begin{align}
(\beta+p)\widehat{\phi}(p)-(\zeta+p)\dot{\phi}(p)+(h-f)\phi(p) =0,\label{L3}
\end{align}
which composes a triplet together with \eqref{L1} and \eqref{L2}.
Introducing the transformation
\begin{eqnarray}
\varphi(p)= (\alpha+p)^n(\beta+p)^m(\zeta+p)^l\phi(p),
\end{eqnarray}
the  triplet   takes the   form
\begin{subequations}
    \begin{align}
         \widehat{\varphi}(p)-\widetilde{\varphi}(p)+h\varphi(p) &=0,\\
    \dot{\varphi}(p)-\widetilde{\varphi}(p)+f\varphi(p) &=0,\\
    \widehat{\varphi}(p)-\dot{\varphi}(p)+(h-f)\varphi(p) &=0,
    \end{align}
\end{subequations}
whose compatibility requires
\begin{eqnarray}\label{CC}
    \widehat{f}-\widetilde{f}=\dot{h}-\widetilde{h} ,\quad -\dot{h}f+\widehat{f}h=0.
\end{eqnarray}
The system \eqref{CC} is known as the non-potential lattice KP (lKP) equation,
 originally derived by Nimmo \cite{N-JPA-2006},
 and the corresponding scalar form
\begin{align}
    & \widetilde{f}\,\widetilde{\widetilde{f}}\,\widehat{\widetilde{f}}\,
    \dot{\widehat{\widetilde{f}}}\,\dot{\widehat{f}}
- \widetilde{\widetilde{f}}\,\widetilde{\widehat{f}}\,
\dot{\widehat{\widetilde{f}}}\,\dot{\widehat{f}}\,\widehat{f}
+ \widetilde{\widetilde{f}}\,\widehat{\widetilde{f}}\,
\dot{\widehat{\widetilde{f}}}\,\dot{\widehat{f}}\, \dot{f}
- \widetilde{f}\,\widetilde{\widetilde{f}}\,\widehat{\widetilde{\widetilde{f}}}\,
\dot{\widehat{f}}\,\dot{\widetilde{f}}
- \widetilde{f}\,\widetilde{\widetilde{f}}\,\dot{\widehat{\widetilde{f}}}\,
\dot{\widetilde{f}}\,\dot{\widehat{f}}
+ \widehat{f}\,\widehat{\widetilde{f}}\,\dot{\widehat{f}}\,
\dot{\widetilde{f}}\,\widehat{\widetilde{\widetilde{f}}} \nn\\
 +~&
\widetilde{\widetilde{f}}\,\widehat{\widetilde{\widetilde{f}}}\,\widehat{\widetilde{f}}\,
\dot{\widetilde{f}}\,\dot{f}
- \widehat{\widetilde{\widetilde{f}}}\,\widehat{\widetilde{f}}{}^2\,\dot{f}\,\dot{\widetilde{f}}\,
- \widehat{\widetilde{\widetilde{f}}}\,\widehat{\widetilde{f}}\,
\dot{\widehat{f}}\dot{f}\,\dot{\widetilde{f}}\,
+ \widehat{\widetilde{\widetilde{f}}}\,\widetilde{\widehat{f}}\,
\dot{f}\,\dot{\widetilde{f}}{}^{2}\,
+ \widetilde{f}\,\widetilde{\widetilde{f}}\,\dot{\widehat{f}}\,
\dot{\widetilde{f}}\,\dot{\widetilde{\widetilde{f}}}
- \widetilde{\widetilde{f}}\,\widetilde{\widehat{f}}\,\dot{f}\,
\dot{\widetilde{f}}\,\dot{\widetilde{\widetilde{f}}}\,
=0
\label{lattice-KP}
\end{align}
was presented in \cite{F-PhD-2018}. The above non-potential scalar form of the lKP can be written as
\begin{eqnarray}
\widetilde{A} \left( \frac{\dot{B}}{\widetilde{\widetilde{B}}}
- \frac{\widetilde{B}}{\dot{\widetilde{B}}} \right)
= \widetilde{\widetilde{A}} (\dot{B} - \widetilde{B})
+ A \left( \frac{1}{\widetilde{\widetilde{B}}} - \frac{1}{\dot{\widetilde{B}}} \right)
\end{eqnarray}
with
\begin{equation}
 A = \dot{f} - \widehat{f},~~~ B = f / \widehat{f}.
\end{equation}

In addition, substituting \eqref{h} and \eqref{f} into \eqref{CC},
the first equation in \eqref{CC} is satisfied automatically,
while the second equation, after defining
\begin{equation}\label{u-phi}
u=\phi^{(-1)},
\end{equation}
gives rise to the lpKP equation \eqref{dKP}.

\subsection{Eigenfunction equations of the lKP type}\label{sec-2-3}

\subsubsection{The lmKP equation}\label{sec-2-3-1}

To construct eigenfunction equations, we consider a point $c\in \mathbb{C}$, where $\phi(p)$ is analytic.
Then, in a neighbourhood of $c$, $\phi(p)$ can be expanded as the following:
\begin{eqnarray}\label{P-EXP}
    \phi(p)=\sum_{\iota=0}^{+\infty}\phi_c^{(\iota)}(p-c)^{\iota},
\end{eqnarray}
where the coefficients can be uniquely expressed as
\begin{align}
\phi_c^{\left(0\right)}&= \phi(p=c)=1+\frac{1}{2\pi {\rm i}}\int _\mathbb{C}\int _\mathbb{C}
\frac{\phi(\mu) R\left( \mu,\xi \right) }{\xi -c} {\rm d}\mu \wedge {\rm d}\overline{\mu }
\wedge {\rm d}\xi \wedge {\rm d}\overline{\xi },\\
    \phi_c^{\left(\iota\right)}&= \frac{1}{2\pi {\rm i}}\int _\mathbb{C}\int _\mathbb{C}
    \frac{\phi(\mu) R\left( \mu,\xi \right) }{(\xi -c)^{\iota+1}}
    {\rm d}\mu \wedge {\rm d}\overline{\mu }\wedge {\rm d}\xi \wedge {\rm d}\overline{\xi },
    \  \ \text{for}   \ \ \iota\geq 1.
\end{align}
Note that we assume
\begin{equation}
\phi(p,n+j,m+s)
=\sum_{\iota=0}^{+\infty}\phi_c^{(\iota)}(n+j,m+s)(p-c)^{\iota},
\qquad j,s\in\mathbb{Z},
\end{equation}
and such an interchange between the expansion and the shift operators is valid at least for soliton solutions
(see Sec.\ref{sec-2-4}).

Substituting the expansion \eqref{P-EXP} into the Lax pair \eqref{Lax_tri} yields the recurrence system
\begin{subequations}\label{RR}
\begin{align}
&    (\beta+c)\widehat{\phi}_c^{(0)}-(\alpha+c)\widetilde{\phi}_c^{(0)}+h\phi_c^{(0)}=0, \label{RR1}\\
& (\beta+c)\widehat{\phi}_c^{(\iota)}+\widehat{\phi}_c^{(\iota-1)}-(\alpha+c)\widetilde{\phi}_c^{(\iota)}
-\widetilde{\phi}_c^{(\iota-1)}+h\phi_c^{(\iota)}=0,\quad \text{for} \quad \iota\geq 1,
\label{RR2}\\
&  (\zeta+c)\dot{\phi}_c^{(0)}-(\alpha+c)\widetilde{\phi}_c^{(0)}+f\phi_c^{(0)}=0,\label{RR3}\\
&(\zeta+c)\dot{\phi}_c^{(\iota)}+\dot{\phi}_c^{(\iota-1)}- (\alpha+c)\widetilde{\phi}_c^{(\iota)}-\widetilde{\phi}_c^{(\iota-1)}+f\phi_c^{(\iota)}=0,
\quad \text{for} \quad \iota \geq 1,\label{RR4}
\end{align}
\end{subequations}
where \eqref{RR1} and \eqref{RR3}
can be written as the following forms
\begin{subequations}\label{T}
\begin{align}
     \alpha-\beta+\widetilde{\phi}^{(-1)}-\widehat{\phi}^{(-1)}=h
     &=\frac{(\alpha+c)\widetilde{\phi}_c^{(0)}-(\beta+c)\widehat{\phi}_c^{(0)}}{\phi_c^{(0)}}, \label{T1}\\
   \alpha-\zeta+\widetilde{\phi}^{(-1)}-\dot{\phi}^{(-1)}=f
   &=\frac{(\alpha+c)\widetilde{\phi}_c^{(0)}-(\zeta+c)\dot{\phi}_c^{(0)}}{\phi_c^{(0)}}. \label{T2}
\end{align}
\end{subequations}
Now, substituting the above $h$ and $f$  into the second equation in the compatibility conditions \eqref{CC},
we obtain an equation only containing  $\phi_c^{(0)}$:
\begin{eqnarray}\label{DMKP_c}
    \frac{(\alpha+c)\dot{\widetilde{\phi}}_c^{(0)}
    -(\beta+c)\dot{\widehat{\phi}}_c{}^{(0)}}{\dot{\phi}_c^{(0)}}
    +\frac{(\beta+c)\widetilde{\widehat{\phi}}_c\!\!{}^{(0)}
    -(\zeta+c)\dot{\widetilde{\phi}}_c^{(0)}}{\widetilde{\phi}_c^{(0)}}
    +\frac{(\zeta+c)\dot{\widehat{\phi}}_c^{(0)}-(\alpha+c)\widehat{\widetilde{\phi}}_c\!\!{}^{(0)}}
    {\widehat{\phi}_c^{(0)}}=0.
\end{eqnarray}
Setting a new variable by
\begin{eqnarray}\label{w-phi}
   w=\phi_c^{(0)},
\end{eqnarray}
Eq. \eqref{DMKP_c} is rewritten as
\begin{eqnarray}\label{dMKP}
    \frac{(\alpha+c) \dot{\widetilde{w}}-(\beta+c) \dot{\widehat{w}}}{\dot{w}}
    +\frac{(\beta+c)\widetilde{\widehat{w}}-(\zeta+c)\dot{\widetilde{w}}}{\widetilde{w}}
    +\frac{(\zeta+c)\dot{\widehat{w}}-(\alpha+c)\widehat{\widetilde{w}}}{\widehat{w}}=0.
\end{eqnarray}
This is the lattice modified KP (lmKP) equation with a parameter $c$.
And the coupled system \eqref{T1} and \eqref{T2}, i.e.
\begin{subequations}
\begin{align}
& \alpha-\beta+\widetilde{u}-\widehat{u}=\frac{(\alpha+c)\widetilde{w}-(\beta+c) \widehat{w}}{w}, \\
&   \alpha-\zeta+\widetilde{u}-\dot{u}=\frac{(\alpha+c) \widetilde{w}-(\zeta+c) \dot{w}}{w}
\end{align}
\end{subequations}
provides the Miura  transformation between the lpKP equation and the lmKP equation \cite{HJN-book-2016}.
Thus, the eigenfunction of the Lax pair \eqref{Lax_tri} for the lKP (lpKP) equation
defines the lmKP equation.

\subsubsection{The lSKP equation and NQC-KP equation}\label{sec-2-3-2}

In addition to $\phi(c)=\phi_c^{(0)}$ that gives rise to the lmKP equation, the higher-order coefficient $\phi_c^{(1)}$ in \eqref{P-EXP} can also be involved to generate integrable lattice equations.
In what follows, we focus on the  construction of such integrable equations.

Eliminating $h$ from  \eqref{RR1} and  \eqref{RR2}$|_{\iota=1}$,
and eliminating $f$ from  \eqref{RR3} and  \eqref{RR4}$|_{\iota=1}$,
we obtain
\begin{subequations}
\begin{align}
(\beta+c)(\widehat{\phi}_c^{(1)}\phi_c^{(0)}-\widehat{\phi}_c^{(0)}\phi_c^{(1)})
+\widehat{\phi}_c^{(0)}\phi_c^{(0)}-(\alpha+c)(\widetilde{\phi}_c^{(1)}\phi_c^{(0)}
-\widetilde{\phi}_c^{(0)}\phi_c^{(1)})-\widetilde{\phi}_c^{(0)}\phi_c^{(0)}&=0,\\
(\zeta+c)(\dot{\phi}_c^{(1)}\phi_c^{(0)}-\dot{\phi}_c^{(0)}\phi_c^{(1)})
+\dot{\phi}_c^{(0)}\phi_c^{(0)}-(\alpha+c)(\widetilde{\phi}_c^{(1)}\phi_c^{(0)}
-\widetilde{\phi}_c^{(0)}\phi_c^{(1)})-\widetilde{\phi}_c^{(0)}\phi_c^{(0)}&=0,
\end{align}
\end{subequations}
which we rewrite as
\begin{subequations}\label{E3-1}
\begin{align}
         (\beta+c)\bigg(\frac{\widehat{\phi}_c^{(1)}}{\widehat{\phi}_c^{(0)}}
         -\frac{\phi_c^{(1)}}{\phi_c^{(0)}}\bigg)
         +1-(\alpha+c)\bigg(\frac{\widetilde{\phi}_c^{(1)}}{\widehat{\phi}_c^{(0)}}
         -\frac{\widetilde{\phi}_c^{(0)}\phi_c^{(1)}}{\widehat{\phi}_c^{(0)}\phi_c^{(0)}}\bigg)
         &=\frac{\widetilde{\phi}_c^{(0)}}{\widehat{\phi}_c^{(0)}},\\
  (\zeta+c)\bigg(\frac{\dot{\phi}_c^{(1)}}{\dot{\phi}_c^{(0)}} -\frac{\phi_c^{(1)}}{\phi_c^{(0)}}\bigg)
  +1-(\alpha+c)\bigg(\frac{\widetilde{\phi}_c^{(1)}}{\dot{\phi}_c^{(0)}}
  -\frac{\widetilde{\phi}_c^{(0)}\phi_c^{(1)}}{\dot{\phi}_c^{(0)}\phi_c^{(0)}}\bigg)
  &=\frac{\widetilde{\phi}_c^{(0)}}{\dot{\phi}_c^{(0)}}.
\end{align}
\end{subequations}
In terms of  a new variable
\begin{eqnarray}\label{z_c-phi}
    z_c=\frac{\phi_c^{(1)}}{\phi_c^{(0)}}+\frac{n}{\alpha+c}+\frac{m}{\beta+c}+\frac{l}{\zeta+c},
\end{eqnarray}
they are written as
\begin{subequations}
\label{2.36}
\begin{eqnarray}\label{T_SKPandMKP}
    \frac{\widetilde{w}}{\widehat{w}}=\frac{ (\beta+c)(\widehat{z}_c-z_c)}{(\alpha+c)(\widetilde{z}_c-z_c)},
    \quad   \frac{\widetilde{w}}{\dot{w}}
    =\frac{ (\zeta+c)(\dot{z}_c-z_c)}{(\alpha+c)(\widetilde{z}_c-z_c)}.
\end{eqnarray}
In addition, as a consequence, we have
\begin{equation}
    \quad   \frac{\dot{w}}{\widehat{w}}
    =\frac{ (\beta+c)(\widehat{z}_c-z_c)}{(\zeta+c)(\dot{z}_c-z_c)}.
\end{equation}
\end{subequations}
Consider the identity
\begin{eqnarray}\label{2.37}
E_m{\Big(\frac{\widetilde{w}}{\dot{w}}\Big) } \cdot E_n{\Big(\frac{\dot{w}}{\widehat{w}}\Big)}=E_l{\Big(\frac{\widetilde{w}}{\widehat{w}}\Big)} ,
\end{eqnarray}
where $E_n, E_m, E_l$ are shift operators defined by acting $f=f(n,m,l)$ as
\begin{equation}
E_n f=\widetilde f,~~  E_m f= \widehat f,~~ E_l f= \dot f .
\end{equation}
Inserting \eqref{2.36} into  \eqref{2.37}, we obtain the discrete Schwarzian KP (lSKP) equation
\begin{eqnarray}\label{dSKP}
    \frac{(\dot{\widehat{z}}_c-\widehat{z}_c)(\widetilde{\widehat{z}}_c
    -\widetilde{z}_c)(\dot{\widetilde{z}}_c-\dot{z}_c)}
    {(\dot{\widehat{z}}_c-\dot{z}_c)(\widetilde{\widehat{z}}_c
    -\widehat{z}_c)(\dot{\widetilde{z}}_c-\widetilde{z}_c)}=1.
\end{eqnarray}
The equations in \eqref{2.36} can be viewed as a set of B\"acklund transformations
between the lmKP equation \eqref{dMKP} and the above lSKP equation \eqref{dSKP}.

In the following, let us consider $c=a$ in the above formulation and introduce a function $Z(a,b)$ by
\begin{eqnarray}\label{z-a}
    z_a=\theta(a,b)\big[(b-a)Z(a,b)+1\big],
\end{eqnarray}
i.e.
\begin{equation}\label{Z-ab}
Z(a,b)=\frac{1}{b-a}(z_a \theta^{-1}(a,b)-1),
\end{equation}
where
\begin{eqnarray}\label{theta}
    \theta(a,b)=\frac{\rho(b)}{\rho(a)}=\bigg(\frac{\alpha+b}{\alpha+a}\bigg)^n
    \bigg(\frac{\beta+b}{\beta+a}\bigg)^m
    \bigg(\frac{\zeta+b}{\zeta+a}\bigg)^l.
\end{eqnarray}
Then, substituting \eqref{z-a} into the lSKP equation \eqref{dSKP}$|_{c=a}$
yields the following equation in terms of $Z(a,b)$:
\begin{eqnarray}
  \frac{ (\beta+b)\widetilde{\widehat{Z}}(a,b)-(\beta+a)\widetilde{Z}(a,b)+1}
  {(\zeta+b) \dot{\widetilde{Z}}(a,b)-(\zeta+a)\widetilde{Z}(a,b)+1}
  \cdot\frac{ (\zeta+b)\dot{\widehat{Z}}(a,b)-(\zeta+a)\widehat{Z}(a,b)+1}
  {(\alpha+b)\widehat{\widetilde{Z}}(a,b)-(\alpha+a)\widehat{Z}(a,b)+1}\nn\\ \cdot\frac{(\alpha+b)\dot{\widetilde{Z}}(a,b)-(\alpha+a)\dot{Z}(a,b)+1}
  {(\beta+b)\dot{\widehat{Z}}(a,b)-(\beta+a)\dot{Z}(a,b)+1}=1,\label{Zab}
\end{eqnarray}
which is known as the NQC-KP equation \cite{HJN-book-2016}.
There are more than one ways to formulate this equation. We will introduce a second one in the following.

\subsubsection{A second formulation of the NQC-KP equation}\label{sec-2-3-3}

We consider the eigenfunction $\phi(p)$ which is defined as a solution of the Dbar problem \eqref{DP-1}
and has been shown to satisfy the eigenvalue problems \eqref{L1}, \eqref{L2} and \eqref{L3}.
Let $\phi(a)$ and $\phi(b)$ be two eigenfunctions by taking $p=a,b$ in $\phi(p)$, respectively, where we assume that $\phi(a)\neq0$.
Define
\begin{equation}\label{X-ab}
X(a,b)=\frac{\phi(b)}{\phi(a)}
\end{equation}
and introduce the function
\begin{eqnarray}\label{S-ab}
    \mathcal{S}(a,b)=\frac{1}{b-a} (X(a,b)-1 ).
\end{eqnarray}
In the following, we will see that such a function provides a solution to the NQC-KP equation as well.

Taking $p=a,b$ in  \eqref{L1} gives
\begin{subequations}
\begin{align}
  &     (\beta+a)\widehat{\phi}(a)-(\alpha+a)\widetilde{\phi}(a)+h\phi(a)  = 0,\\
  &     (\beta+b)\widehat{\phi}(b)-(\alpha+b)\widetilde{\phi}(b)+h\phi(b) = 0.
\end{align}
\end{subequations}
Noting that $h$ is the same in the two above equations since it is independent of $p$,
we eliminate $h$ and obtain
\begin{eqnarray}
       (\beta+b)\widehat{\phi}(b)\phi(a)-(\alpha+b)\widetilde{\phi}(b)\phi(a)
       +(\alpha+a)\widetilde{\phi}(a)\phi(b)- (\beta+a)\widehat{\phi}(a)\phi(b)=0,
   \end{eqnarray}
which can be rewritten in terms of $X(a,b)$ as
\begin{subequations}
\begin{eqnarray}
       \frac{(\beta+b)\widehat{X}(a,b)-(\beta+a)X(a,b)}{(\alpha+b)\widetilde{X}(a,b)-(\alpha+a)X(a,b)}
=\frac{\widetilde{\phi}(a)}{\widehat{\phi}(a)}.
   \end{eqnarray}
Similarly, from \eqref{L2} and \eqref{L3} we have
\begin{align}
&   \frac{(\zeta+b)\dot{X}(a,b)-(\zeta+a)X(a,b)}{(\alpha+b)\widetilde{X}(a,b)-(\alpha+a)X(a,b)}
=\frac{\widetilde{\phi}(a)}{\dot{\phi}(a)},
\\
&\frac{(\beta+b)\widehat{X}(a,b)-(\beta+a)X(a,b)}{(\zeta+b)\dot{X}(a,b)-(\zeta+a)X(a,b)}
=\frac{\dot{\phi}(a)}{\widehat{\phi}(a)}.
\end{align}
\end{subequations}
These three relations, together with the identity
\begin{eqnarray}
E_l \bigg(\frac{\widetilde{\phi}(a)}{\widehat{\phi}(a)}\bigg)=
E_m \bigg(\frac{\widetilde{\phi}(a)}{\dot{\phi}(a)}\bigg)
\cdot E_n \bigg(\frac{\dot{\phi}(a)}{\widehat{\phi}(a)}\bigg)
\end{eqnarray}
yield an equation for $X(a,b)$:
\begin{eqnarray}\label{Xab}
    \frac{(\beta+b)\dot{\widehat{X}}(a,b)-(\beta+a)\dot{X}(a,b)}{(\alpha+b)\dot{\widetilde{X}}(a,b)
    -(\alpha+a)\dot{X}(a,b)}\cdot\frac{(\alpha+b)\widehat{\widetilde{X}}(a,b)-(\alpha+a)\widehat{X}(a,b)}
    {(\zeta+b)\dot{\widehat{X}}(a,b)-(\zeta+a)\widehat{X}(a,b)}\nn\\ \cdot\frac{(\zeta+b)\dot{\widetilde{X}}(a,b)-(\zeta+a)\widetilde{X}(a,b)}
    {(\beta+b)\widehat{\widetilde{X}}(a,b)-(\beta+a)\widetilde{X}(a,b)}
=1.
\end{eqnarray}
Meanwhile, from \eqref{S-ab} we have
\begin{eqnarray}
    X(a,b)=(b-a)\mathcal{S}(a,b)+1.
\end{eqnarray}
Inserting it into \eqref{Xab} yields the NQC-KP equation
\begin{eqnarray}
    \frac{(\beta+b)\dot{\widehat{\mathcal{S}}}(a,b)-(\beta+a)\dot{\mathcal{S}}(a,b)+1}
    {(\alpha+b)\dot{\widetilde{\mathcal{S}}}(a,b)
    -(\alpha+a)\dot{\mathcal{S}}(a,b)+1}\cdot
    \frac{(\alpha+b)\widehat{\widetilde{\mathcal{S}}}(a,b)-(\alpha+a)\widehat{\mathcal{S}}(a,b)+1}
    {(\zeta+b)\dot{\widehat{\mathcal{S}}}(a,b)-(\zeta+a)\widehat{\mathcal{S}}(a,b)+1}\nn\\ \cdot\frac{(\zeta+b)\dot{\widetilde{\mathcal{S}}}(a,b)-(\zeta+a)\widetilde{\mathcal{S}}(a,b)+1}
    {(\beta+b)\widehat{\widetilde{\mathcal{S}}}(a,b)-(\beta+a)\widetilde{\mathcal{S}}(a,b)+1}
=1.
\label{NQC-2}
\end{eqnarray}
This provides a second formulation of the NQC-KP equation (see Sec.\ref{sec-3-4}).

\subsection{$N$-soliton solutions}\label{sec-2-4}

We have constructed the lKP equation \eqref{lattice-KP}, the lpKP equation \eqref{dKP},
and their eigenfunction equations, namely,
the lmKP equation \eqref{dMKP}, the lSKP equation \eqref{dSKP},
the NQC-KP equation \eqref{Zab} in terms of $Z(a,b)$
and the NQC-KP equation \eqref{NQC-2} in terms of $\mathcal{S}(a,b)$.
Their solutions are formulated through the eigenfunction $\phi(p)$ respectively by
\eqref{f}, \eqref{u-phi}, \eqref{w-phi}, \eqref{z_c-phi}, \eqref{Z-ab} and \eqref{S-ab}.
In this framework, once we get explicit expression for $\phi$ by
solving the $\bar{\partial}$ problem \eqref{DP-1},
or solving the integral equation \eqref{LI} (in light of Theorem \ref{thm-1}),
we may get multi-soliton solutions for all the above mentioned  equations of the lKP-type.

To get explicit forms of $\phi(p)$ and the coefficients $\phi_c^{\iota}$
we consider the following special choices,  for which the matrix $\bI+\mathcal{K}$ is generically invertible,
\begin{eqnarray}\label{R0}
r_1^0(p)
=
\begin{pmatrix}
\delta(p-k_1)\\
\delta(p-k_2)\\
\vdots\\
\delta(p-k_N)
\end{pmatrix},~\quad
r_2^0(p)
=2\pi {\mathrm i}
\begin{pmatrix}
\rho_{1,1}^{(0)}&\rho_{1,2}^{(0)}&\cdots&\rho_{1,N'}^{(0)}\\
\rho_{2,1}^{(0)}&\rho_{2,2}^{(0)}&\cdots&\rho_{2,N'}^{(0)}\\
\vdots&\ddots&&\vdots\\
\rho_{N,1}^{(0)}&\rho_{N,2}^{(0)}&\cdots&\rho_{N,N'}^{(0)}
\end{pmatrix}
\begin{pmatrix}
\delta(p-k_1')\\
\delta(p-k_2')\\
\vdots\\
\delta(p-k'_{N'})
\end{pmatrix},
\end{eqnarray}
where $\mathrm{i}$ is the imaginary unit, $ k_j,k_{\ell}^{'}, \rho^{(0)}_{j,\ell} \in \mathbb{C}$ and
\begin{eqnarray*}
    k_j+k_{\ell}^{'}\neq 0,~k_j\neq k_{\ell}^{'},~ k_{\ell}^{'}\neq -\alpha,-\beta,-\zeta, ~~~ \forall j,\ell\in \mathbb{N}.
\end{eqnarray*}
Here, the $\delta$-function on the complex plane $\mathbb{C}$ (see Appendix A in \cite{SZZ-PA-2026} for more details) is defined as the following,
\begin{eqnarray}
    \int_{\Omega} f(p) \, \delta(p - p_0) \, {\rm d}p\wedge{\rm d}\overline{p} = f(p_0),
\end{eqnarray}
where the domain $\Omega\subset \mathbb{C}$ contains the point $p_0$ and $f(p)$ is sufficiently smooth in $\Omega$.
Substituting \eqref{R0} into \eqref{R} and then into \eqref{LI} leads to
\begin{eqnarray}
\phi( p)
&=&1+\sum_{j=1}^{N}\sum_{\ell=1}^{N'}\frac{1}{2\pi {\rm i}}\frac{\phi(k_j)\rho_{j,\ell}^{(0)}
\left(\frac{\alpha+k_j}{\alpha+k_{\ell}^{'}}\right)^n
\left(\frac{\beta+k_j}{\beta+k_{\ell}^{'}}\right)^m
\left(\frac{\zeta+k_j}{\zeta+k_{\ell}^{'}}\right)^l
2\pi {\rm i}}{k_{\ell}^{'}-p}\nonumber \\
&=&  1+\sum ^{N}_{j=1}\sum_{\ell=1}^{N'}\frac{\phi (k_j) \rho(k_j)\rho^{-1}(k_{\ell}^{'})\rho_{j,\ell}^{(0)}}{k_{\ell}^{'}-p},
\label{NS_1}
\end{eqnarray}
where (see \eqref{rho})
\begin{equation}\label{rho-k}
\rho(k)=(\alpha+k)^n(\beta+k)^m(\zeta+k)^l.
\end{equation}
Then, taking $p=k_1, k_2, \cdots, k_N$ in \eqref{NS_1}, we obtain an equation set
\begin{eqnarray}
    \phi \left(k_{i}\right)+\sum ^{N}_{j=1}\sum_{\ell=1}^{N'}\frac{\phi (k_j) \rho(k_j)\rho^{-1}(k_{\ell}^{'})\rho_{j,\ell}^{(0)}}{k_i-k_{\ell}^{'}}=1,
    ~~~ (i=1,2,\cdots,N),
\end{eqnarray}
which can be rewritten into the matrix form
\begin{subequations}
\begin{eqnarray}
       (\bI_N+\bM\bu^{-1}_{N'} \bB \bv_{N})\Phi=\bi_N,
\end{eqnarray}
where $\bI_N$ denotes the $N\times N$ identity matrix,
\begin{align}
& \bM=\left(M_{i,j}\right)_{N\times N'},~ M_{i,j}=\frac{1}{k_i-k_j^{'}},~~~
   \bB=\left(B_{\ell,j}\right)_{N'\times N}, ~ B_{\ell,j}=\rho^{(0)}_{j,\ell},\label{2.59b}\\
& \bu_{N'}=\mathrm{Diag}\{\rho(k'_1),\rho({k'_2}),\cdots,\rho(k'_{N'})\}, ~~~
   \bv_N=\mathrm{Diag}\{\rho(k_1),\rho(k_2),\cdots,\rho(k_{N})\},  \label{2.59c}\\
& \Phi=\left( \phi(k_1),  \phi(k_2), \cdots,  \phi(k_N)\right)^T,~~~
\bi_N=\underbrace{ (1, 1, \cdots  1)}_{\text{N-dimensional}}\!\!{}^T. \label{2.59d}
\end{align}
\end{subequations}
Thus we have
\begin{equation}\label{Phi}
\Phi=  (\bI_N+\bM\bu^{-1}_{N'} \bB \bv_{N})^{-1}\bi_N.
\end{equation}
In addition, $\phi(p)$ given by \eqref{NS_1} can be written as
\begin{subequations}
\begin{align}
    \phi(p)=&1+\bi_{N'}^T (\bK'-p\bI_{N'})^{-1} \bu^{-1}_{N'} \bB \bv_{N}
    \left(\bI_N+\bM\bu^{-1}_{N'} \bB \bv_{N}\right)^{-1}\bi_N\nn \\
    =&1+\bi_{N'}^T \bu^{-1}_{N'}(\bK'-p\bI_{N'})^{-1}  \bB \left(\bI_N+\bv_{N}\bM\bu^{-1}_{N'} \bB \right)^{-1}
    \bv_{N}\bi_N\nn \\
=&1+\bt^T(\bK'-p\bI_{N'})^{-1} \bB \left(\bI_N+\bv_{N}\bM\bu^{-1}_{N'} \bB \right)^{-1}\bs,
\label{2.61a}
\end{align}
where
\begin{align}
& \bt=( \rho^{-1}(k'_1),  \rho^{-1}(k'_2),  \cdots, \rho^{-1}(k'_{N'}))^T,~~~
\bs=(\rho(k_1), \rho(k_2), \cdots, \rho(k_{N}))^T, \label{2.61b}\\
& \bK'=\mathrm{Diag}\{k_1^{'},k_2^{'},\cdots, k_{N'}^{'}\}. \label{2.61c}
\end{align}
\end{subequations}
Obviously, from \eqref{LI} we also have
\begin{subequations}
    \begin{align}
            \phi ^{(\iota)}&=-\frac{1}{2\pi {\rm i}}
      \int_\mathbb{C} \frac{\phi\left( \mu \right) R\left( \mu ,\xi\right)}{ \xi^{\iota+1} }  d\mu \wedge d\bar{\mu },
      \quad \text{for} \quad \iota\leq-1, \\\phi_c^{\left(0\right)}&=1+\frac{1}{2\pi {\rm i}}\int _\mathbb{C}\int _\mathbb{C}\frac{\phi(\mu) R\left( \mu,\xi \right) }{\xi -c} {\rm d}\mu \wedge {\rm d}\overline{\mu }\wedge {\rm d}\xi \wedge {\rm d}\overline{\xi },\\
    \phi_c^{\left(\iota\right)}&=\frac{1}{2\pi {\rm i}}\int _\mathbb{C}\int _\mathbb{C}\frac{\phi(\mu) R\left( \mu,\xi \right) }{(\xi -c)^{\iota+1}} {\rm d}\mu \wedge {\rm d}\overline{\mu }\wedge {\rm d}\xi \wedge {\rm d}\overline{\xi },\  \ \text{for}   \ \ \iota \geq 1.
    \end{align}
\end{subequations}
Then, by a way similar to getting \eqref{2.61a}, we obtain
\begin{subequations}
\begin{align}
        \phi ^{(\iota)}&=-\bt^T \bK'^{-(\iota+1)}\bB \left(\bI_N+\bv_{N}\bM\bu^{-1}_{N'} \bB \right)^{-1}\bs.\ \ \text{for}  \  \  \iota\leq-1,\\
        \phi_c^{(0)}&=\phi(c)=1+\bt^T (\bK'-c\bI_{N'})^{-1} \bB \left(\bI_N+\bv_{N}\bM\bu^{-1}_{N'} \bB \right)^{-1}\bs, \label{phi-c0}\\
    \phi_c^{(\iota)}&=\bt^T (\bK'-c\bI_{N'})^{-(\iota+1)} \bB \left(\bI_N+\bv_{N}\bM\bu^{-1}_{N'} \bB \right)^{-1}\bs,\  \ \text{for}  \ \ \iota\geq1.
    \end{align}
\end{subequations}

Thus, we have obtained all the explicit formulae for $\phi(p)$ and the coefficients \textcolor{red}{$\phi_c^{\iota}$}.
We summarize solution formulations of the obtained lKP-type equations in Table \ref{tab:my_table}.
\begin{table}[H]
\begin{tabular}{|l|l|>{\raggedright\arraybackslash}p{0.3\linewidth}|}
\hline
{Equation}   & {Variable} &   {Solution} \\[0.7ex]
\hline
lKP \eqref{lattice-KP} & $f$ &     $\alpha-\zeta+\widetilde{\phi}^{(-1)}-\dot{\phi}^{(-1)}$  \\[0.7ex]
\hline
lpKP \eqref{dKP}& $u$&     $\phi^{(-1)}$  \\[0.7ex]
\hline
lmKP \eqref{dMKP}&$w$&   $ \phi_c^{(0)} $\\[0.7ex]
\hline
lSKP \eqref{dSKP}& $z_c$&   $\frac{\phi_c^{(1)}}{\phi_c^{(0)}}+\frac{n}{\alpha+c}+\frac{m}{\beta+c}+\frac{l}{\zeta+c}$\\[0.7ex]
\hline
NQC-KP \eqref{Zab}&$Z(a,b)$&    $\Xi$\\[0.7ex]
\hline
NQC-KP { \eqref{NQC-2}}&$\mathcal{S}(a,b)$&     $\frac{1}{b-a}\Big(\frac{\phi(b)}{\phi(a)}-1\Big)$\\[0.7ex]
\hline
\end{tabular}
\vspace{12pt}
\centering
\caption{Solutions of the lKP-type equations}

\label{tab:my_table}
\end{table}

\noindent
where
\begin{subequations}
\begin{eqnarray}
    \phi ^{(-1)}&=&-\bt^T \bB \left(\bI_N+\bv_{N}\bM\bu^{-1}_{N'} \bB \right)^{-1}\bs,\\
    \phi_c^{(0)}&=&1+\bt^T (\bK'-c\bI_{N'})^{-1} \bB \left(\bI_N+\bv_{N}\bM\bu^{-1}_{N'} \bB \right)^{-1}\bs, \\
    \phi_c^{(1)}&=&\bt^T (\bK'-c\bI_{N'})^{-2} \bB \left(\bI_N+\bv_{N}\bM\bu^{-1}_{N'} \bB \right)^{-1}\bs,\\
    \phi(p)&=&1+\bt^T (\bK'-p\bI_{N'})^{-1} \bB \left(\bI_N+\bv_{N}\bM\bu^{-1}_{N'} \bB \right)^{-1}\bs,\\
    \Xi&=&\frac{1}{b-a}\left(\theta^{-1}\left(\frac{\phi_a^{(1)}}{\phi_a^{(0)}}
    +\frac{n}{\alpha+a}+\frac{m}{\beta+a}+\frac{l}{\zeta+a}\right)-1\right),
    \end{eqnarray}
\end{subequations}
and $\theta=\theta(a,b)$ is given in \eqref{theta}.

\section{Inhomogeneous nonlocal $\bar{\partial}$ problem and related equations}\label{sec-3}

In this section, we will consider an inhomogeneous $\bar{\partial}$ problem,
of which the solution not only provides the eigenfunction of the Lax pair of the lmKP equation,
but also directly leads to a formulation of the NQC-KP equation.

\subsection{Inhomogeneous $\bar{\partial}$ problem, Lax pair and lmKP equation}\label{sec-3-1}

Let us consider an inhomogeneous $\bar{\partial}$ problem
\begin{subequations}\label{DP-2}
\begin{eqnarray}\label{Dp-d}
    \bar{\partial} \psi\left( p\right) =-2\pi{\rm i}\delta(p-b)+\int_C\psi \left( \mu\right) R\left(\mu, p\right){\rm d}\mu\wedge {\rm d}\bar{\mu}
\end{eqnarray}
with a special asymptotic behavior
\begin{equation}
    \psi \left( p\right) \sim 0, ~~\  \ (p\rightarrow \infty),
\end{equation}
\end{subequations}
where $p$ is the spectral parameter, $b\in \mathbb{C}$ is a parameter,
and $R(\mu,p)$ is  defined in \eqref{R} and \eqref{rho}.

This is different from the homogeneous $\bar{\partial}$ problem \eqref{DP-1}
because of  the additional inhomogeneous term, i.e.  $-2\pi{\rm i}\delta(p-b)$.
Through the Cauchy-Pompeiu integral formula \eqref{Cauchy}, the $\bar{\partial}$ problem \eqref{DP-2} is cast to
\begin{equation}\label{LI2}
    \psi \left( p\right) =\frac{1}{p-b}+\frac{1}{2\pi {\rm i}}\int _\mathbb{C}\int _\mathbb{C}\frac{\psi(\mu) R\left( \mu,\xi \right) }{\xi -p} {\rm d}\mu \wedge {\rm d}\overline{\mu }\wedge {\rm d}\xi
    \wedge {\rm d}\overline{\xi }.
\end{equation}
In the following analysis, we assume $\psi(p)$ has the following Laurent expansion at infinity:
\begin{equation}\label{exp-infty}
\psi \left( p\right) =\sum_{j=1}^{+\infty}b^{j-1}p^{-j}+
\sum_{j=1}^{+\infty} \psi^{(-j)} p^{-j},
\end{equation}
where the two series correspond to the two terms on the right-hand side of \eqref{LI2}, respectively.
Following the same procedure as in Sec.\ref{sec-2-2},
one can find that the function $\psi$ defined by \eqref{DP-2} satisfies the Lax triplet
\begin{subequations}\label{Lax-MKP}
\begin{align}
    (\beta+p)\widehat{\psi}(p)-(\alpha+p) g_1\widetilde{\psi}(p)
+(\alpha  g_1-\beta +  bg_1-b)\psi(p)&=0,\label{Lax-MKP-1}\\
(\zeta+p)\dot{\psi}(p)-(\alpha+p) g_2\widetilde{\psi}(p)
+(\alpha  g_2-\zeta +  bg_2-b)\psi(p)&=0,\label{Lax-MKP-2}\\
(\beta+p)\widehat{\psi}(p)-(\zeta+p) g_3\dot{\psi}(p)
+(\zeta g_3-\beta +  bg_3-b)\psi(p)&=0,\label{Lax-MKP-3}
\end{align}
\end{subequations}
where
\begin{eqnarray}\label{Define2}
    g_1=\frac{1+\widehat{\psi}^{(-1)}}{1+\widetilde{\psi}^{(-1)}},~~~
     g_2=\frac{1+\dot{\psi}^{(-1)}}{1+\widetilde{\psi}^{(-1)}},~~~
      g_3=\frac{1+\widehat{\psi}^{(-1)}}{1+\dot{\psi}^{(-1)}},
\end{eqnarray}
and note that they are connected by
\begin{equation}\label{g123}
\dot g_1=\widehat g_2\, \widetilde g_3.
\end{equation}
Here we provide the details for the derivation of \eqref{Lax-MKP-1}.
First, in light of \eqref{exp-infty}, we have
\begin{align}
 N(p)\psi(p)&:= (\beta+p)\widehat{\psi}(p)-(\alpha+p) g_1\widetilde{\psi}(p)
+(\alpha  g_1-\beta +  bg_1-b)\psi(p)\nonumber\\
&=(1+\widehat{\psi}^{(-1)})-g_1(1+\widetilde{\psi}^{(-1)})+O(p^{-1})\nonumber\\
&= 0+O(p^{-1}),~~~~~~~~~~~~~~~~~~~~~~~~~~~~~~~~~~~~~~~~~~~
\text{as}~~|p| \to \infty.
\label{Np}
\end{align}
Then we calculate
\begin{equation}
  \bar{\partial}\big(N(p)\psi(p)\big)
  = (\beta+p)\bar{\partial}\widehat{\psi}(p)-(\alpha+p) g_1\bar{\partial}\widetilde{\psi}(p)
+(\alpha  g_1-\beta +  bg_1-b)\bar{\partial}\psi(p).
\end{equation}
After inserting  the $\bar{\partial}$ equation \eqref{Dp-d} into the right-hand side,
we get a homogeneous  $\bar\partial$ problem
\begin{equation}
  \bar{\partial}\big(N(p)\psi(p)\big)
  =\int_C\big(N(\mu)\psi(\mu)\big) R\left(\mu, p\right){\rm d}\mu\wedge {\rm d}\bar{\mu}.
\end{equation}
Thus, in light of Theorem \ref{thm-1} and the asymptotic result \eqref{Np}, we have
$N(p)\psi(p)=0$, which means \eqref{Lax-MKP-1} holds.
In a similar way we can prove \eqref{Lax-MKP-2}.
The third relation \eqref{Lax-MKP-3} is the driect consequence of  \eqref{Lax-MKP-1} and \eqref{Lax-MKP-2}.

The set \textcolor{red}{of} equations \eqref{Lax-MKP} provide  a Lax triplet for the lmKP equation. In fact,
under the transformation
\begin{eqnarray}
    \varphi=(\alpha+p)^n(\beta+p)^m(\zeta+p)^l\psi(p),
\end{eqnarray}
they go to
\begin{subequations}
    \begin{align}
    \widehat{\varphi}(p)- g_1\widetilde{\varphi}(p)
+(\alpha  g_1-\beta +  b g_1-b)\varphi(p)&=0,\\
\dot{\varphi}(p)- g_2\widetilde{\varphi}(p)
+(\alpha  g_2-\zeta +  bg_2-b)\varphi(p)&=0,\\
\widehat{\varphi}(p)- g_3\dot{\varphi}(p)
+(\zeta g_3-\beta +  bg_3-b)\varphi(p)&=0.
    \end{align}
\end{subequations}
The compatibility conditions yield two equations
\begin{subequations}
\begin{align}
         & -\widehat{g}_2(\zeta \widetilde{g}_3-\beta +  b\widetilde{g}_3-b)+(\alpha  \dot{g}_1-\beta
         +  b \dot{g}_1-b)g_2-(\alpha  \widehat{g}_2-\zeta +  b\widehat{g}_2-b)g_1=0,\label{NCC_1}\\
         &(\alpha  \dot{g}_1-\beta +  b\dot{g}_1-b)(\alpha  g_2-\zeta +  bg_2-b)
         =(\alpha  \widehat{g}_2-\zeta +  b\widehat{g}_2-b)(\alpha  g_1-\beta +  b g_1-b), \label{NCC_2}
    \end{align}
\end{subequations}
which are the same in light of the relation \eqref{g123}.
Now we focus on \eqref{NCC_2}.
Introduce
\begin{eqnarray}\label{U}
    U:=1+\psi^{(-1)},
\end{eqnarray}
by which we write \eqref{Define2} as
\begin{eqnarray}\label{Define2-U}
    g_1=\frac{ \widehat{U}}{ \widetilde{U}},~~~
     g_2=\frac{\dot{U}}{\widetilde{U}},~~~
      g_3=\frac{\widehat{U}}{\dot{U}},
\end{eqnarray}
Inserting them into \eqref{NCC_2}, we  rewrite it in terms of $U$ as
\begin{align}
      &  \big((\alpha+b)\dot{\widehat U}-(\beta+b)\dot{\widetilde U}\big)
\big((\alpha+b)\dot U - (\zeta+b)\widetilde U\big)\,\widehat{\widetilde U}\nn \\
=~&
\big((\alpha+b)\dot{\widehat{U}}-(\zeta+b)\widehat{\widetilde U}\big)
\big((\alpha+b)\widehat U - (\beta+b)\widetilde U\big)\,\dot{\widetilde U},
    \end{align}
or equivalently,
\begin{eqnarray}\label{DMKP_U}
    (\alpha+b)\left( \frac{\widehat{U}}{\widehat{\widetilde{U}}}
       -\frac{\dot{U}}{\dot{\widetilde{U}}} \right)
+(\beta+b)\left( \frac{\dot{U}}{\dot{\widehat{U}}}
       -\frac{\widetilde{U}}{\widehat{\widetilde{U}}} \right)
+(\zeta+b)\left( \frac{\widetilde{U}}{\dot{\widetilde{U}}}
       -\frac{\widehat{U}}{\dot{\widehat{U}}} \right)
=0.
\end{eqnarray}
This is the lmKP equation with a parameter $b$, cf.\cite{HJN-book-2016}.

\subsection{The NQC-KP equation}\label{sec-3-2}

In the following we will show that the function $\psi(p)$ defined by the
inhomogeneous $\bar\partial$ problem \eqref{DP-2} satisfies the NQC-KP equation.
There are two parameters $a$ and $b$ in the NQC-KP equation (see \eqref{NQC-2}).
There is already the parameter $b$ in $\psi(p)$.
To introduce the parameter $a$, we assume that $a \in \mathbb{C}$ $ (a \neq b)$
is a finite point where $\psi(p)$ is analytic.
Thus we can  expand $\psi(p)$ in a neighborhood of $a$:
\begin{equation}\label{exp-a}
    \psi (p) =\sum_{\iota=0}^{+\infty}\psi_a^{(\iota)}(p-a)^\iota,
\end{equation}
where $\{\psi_a^{(\iota)}\}$ can be uniquely expressed as
\begin{subequations}
\begin{align}
            \psi_a^{\left(0\right)}&=\psi(a)=\frac{1}{a-b}+\frac{1}{2\pi {\rm i}}
            \int _\mathbb{C}\int _\mathbb{C}\frac{\psi(\mu) R\left( \mu,\xi \right) }{\xi -a}
            {\rm d}\mu \wedge {\rm d}\overline{\mu }\wedge {\rm d}\xi \wedge {\rm d}\overline{\xi },\\
    \psi_a^{\left(\iota\right)}&=-\frac{1}{(b-a)^{\iota+1}}+\frac{1}{2\pi {\rm i}}
    \int _\mathbb{C}\int _\mathbb{C}\frac{\psi(\mu) R\left( \mu,\xi \right) }{(\xi -a)^{\iota+1}}
     {\rm d}\mu \wedge {\rm d}\overline{\mu }\wedge {\rm d}\xi \wedge {\rm d}\overline{\xi },
     \  \ \text{for}   \ \ \iota \geq 1.
    \end{align}
\end{subequations}
Just taking $p=a$ in the triplet \eqref{Lax-MKP} yields
\begin{subequations}\label{Lax-MKP-a}
\begin{align}
    (\beta+a)\widehat{\psi}(a)-(\alpha+a) g_1\widetilde{\psi}(a)
+(\alpha  g_1-\beta +  bg_1-b)\psi(a)&=0,\label{Lax-MKP-a1}\\
(\zeta+a)\dot{\psi}(a)-(\alpha+a) g_2\widetilde{\psi}(a)
+(\alpha  g_2-\zeta +  bg_2-b)\psi(a)&=0,\label{Lax-MKP-a2}\\
(\beta+a)\widehat{\psi}(a)-(\zeta+a) g_3\dot{\psi}(a)
+(\zeta g_3-\beta +  bg_3-b)\psi(a)&=0.\label{Lax-MKP-a3}
\end{align}
\end{subequations}
Define
\begin{eqnarray}\label{SS-ab}
    {S}:={S}(a,b)=\psi(a).
\end{eqnarray}
Then, from \eqref{Lax-MKP-a} we obtain
\begin{equation}\label{g123-S}
g_1 = \frac{(\beta+b)S-(\beta+a)\widehat{S}}
{(\alpha+b)S-(\alpha+a)\widetilde{S}},~~
g_2 = \frac{(\zeta+b)S-(\zeta+a)\dot{S}}
{(\alpha+b)S-(\alpha+a)\widetilde{S}},~~
g_3 = \frac{(\beta+b)S-(\beta+a)\widehat{S}}
{(\zeta+b)S-(\zeta+a)\dot{S}}.
\end{equation}
Now, substituting the above expressions  into the relation \eqref{g123} yields a third formulation
for the NQC-KP equation:
 \begin{eqnarray}\label{Sab-2}
         \frac{(\beta+b)\widetilde{S}-(\beta+a)\widehat{\widetilde{S}}}
         {(\zeta+b)\widetilde{S}-(\zeta+a)\dot{\widetilde{S}}}
         \cdot\frac{(\zeta+b)\widehat{S}-(\zeta+a)\dot{\widehat{S}}}
         {(\alpha+b)\widehat{S}-(\alpha+a)\widehat{\widetilde{S}}}
         \cdot\frac{(\alpha+b)\dot{S}-(\alpha+a)\dot{\widetilde{S}}}
         {(\beta+b)\dot{S}-(\beta+a)\dot{\widehat{S}}}=1.
    \end{eqnarray}
Note that substituting \eqref{g123-S} into \eqref{NCC_2} yields a form
\begin{eqnarray}
\frac{(\beta+b)(\alpha+a)\,\dot{\widetilde{S}} - (\alpha+b)(\beta+a)\,\dot{\widehat{S}}}
{(\beta+b)(\alpha+a)\,\widetilde{S}- (\alpha+b)(\beta+a)\,\widehat{S}}
\cdot \frac{(\zeta+b)(\alpha+a)\,\widetilde{S} - (\alpha+b)(\zeta+a)\,\dot{S}}
{(\zeta+b)(\alpha+a)\,\widetilde{\widehat{S}} - (\alpha+b)(\zeta+a)\,\dot{\widehat{S}}}
\nonumber\\
\cdot
\frac{(\alpha+b)\,\widehat{S}-(\alpha+a)\,\widetilde{\widehat{S}}}
{(\alpha+b)\,\dot{S}-(\alpha+a)\,\dot{\widetilde{S}}}=1,
\end{eqnarray}
which is actually the same as \eqref{Sab-2} after some calculation.
Note also that combining \eqref{Define2-U} and \eqref{g123-S} together yields 
\begin{equation}
\frac{\widehat{U}}{\widetilde{U}}=
\frac{(\beta+b)S-(\beta+a)\widehat{S}}
{(\alpha+b)S-(\alpha+a)\widetilde{S}},~~~
\frac{\dot{U}}{\widetilde{U}}=
 \frac{(\zeta+b)S-(\zeta+a)\dot{S}}
 {(\alpha+b)S-(\alpha+a)\widetilde{S}},
\end{equation}
which provides a set of the B\"acklund transformations
between the lmKP equation \eqref{DMKP_U} and  the NQC-KP equation \eqref{Sab-2}.

\subsection{$N$-soliton solutions}\label{sec-3-3}

We derive the explicit forms of $\psi(p)$ and its coefficients in light of the choice \eqref{R0}.
Substituting  \eqref{R0} into \eqref{LI2} yields
\begin{eqnarray}
\psi( p)
&=&\frac{1}{p-b}+\sum_{j=1}^{N}\sum_{\ell=1}^{N'}\frac{1}{2\pi {\rm i}}
\frac{\psi(k_j)\rho_{j,\ell}^{(0)}\left(\frac{\alpha+k_j}{\alpha+k_{\ell}^{'}}\right)^n
\left(\frac{\beta+k_j}{\beta+k_{\ell}^{'}}\right)^m\left(\frac{\zeta+k_j}{\zeta+k_{\ell}^{'}}\right)^l
2\pi {\rm i}}
{k_{\ell}^{'}-p}\nonumber \\
&=&\frac{1}{p-b}+\sum_{j=1}^{N}\sum_{\ell=1}^{N'}
\frac{\psi (k_j) \rho(k_j)\rho^{-1}(k_{\ell}^{'})\rho_{j,\ell}^{(0)}}{k_{\ell}^{'}-p}.
\label{C2-NS_1}
\end{eqnarray}
Then, taking $p=k_1, k_2, \cdots, k_N$ in \eqref{C2-NS_1} and imposing the restriction $k_i\neq b$, we obtain an equation set
\begin{eqnarray}
    \psi \left(k_{i}\right)+\sum ^{N}_{j=1}\sum_{\ell=1}^{N'}\frac{\psi (k_j) \rho(k_j)\rho^{-1}(k_{\ell}^{'})\rho_{j,\ell}^{(0)}}{k_i-k_{\ell}^{'}}=\frac{1}{k_i-b},
    ~~~ (i=1,2,\cdots,N),
\end{eqnarray}
which can be rewritten as the matrix form
\begin{eqnarray}
       (\bI_N+\bM\bu^{-1}_{N'} \bB \bv_{N})\Psi=(\bK-b\bI_N)^{-1}\bi_N ,
\end{eqnarray}
where $\bM, \bu_{N'}, \bB, \bv_N, \bi_N$ are defined in \eqref{2.59b}, \eqref{2.59c} and \eqref{2.59d}, and
\begin{eqnarray}
\Psi = ( \psi(k_1), \psi(k_2),  \cdots, \psi(k_N))^{T}, ~~
\bK=\mathrm{Diag} \{ k_1, k_2, \cdots, k_N \}.
\end{eqnarray}
Thus we have
\begin{equation}\label{Psi}
\Psi=  (\bI_N+\bM\bu^{-1}_{N'} \bB \bv_{N})^{-1}(\bK-b\bI)^{-1}\bi_N.
\end{equation}
In addition, $\psi(p)$ given in \eqref{C2-NS_1} can be written as
\begin{align}
    \psi(p)=&\frac{1}{p-b}+\bi_{N'}^T (\bK'-p\bI_{N'})^{-1} \bu^{-1}_{N'} \bB \bv_{N}
    \left(\bI_N+\bM\bu^{-1}_{N'} \bB \bv_{N}\right)^{-1}(\bK-b\bI)^{-1}\bi_N\nn
\\
=&\frac{1}{p-b}+\bt^T(\bK'-p\bI_{N'})^{-1}\bB \left(\bI_N+\bv_{N}\bM\bu^{-1}_{N'} \bB \right)^{-1}
(\bK-b\bI)^{-1}\bs,
\end{align}
where $\bt, \bs, \bK'$ are defined in \eqref{2.61b} and \eqref{2.61c}.
Note that $\psi(p)$ has the Laurent expansion \eqref{exp-infty}.
It follows that
\begin{subequations}
\begin{align}
\psi ^{(\iota)}&=-\frac{1}{2\pi {\rm i}}
      \int_\mathbb{C} \int_\mathbb{C} \frac{\psi\left( \mu \right) R\left( \mu ,\xi\right)}{ \xi^{\iota+1} }  d\mu \wedge d\bar{\mu }\wedge {\rm d}\xi \wedge {\rm d}\overline{\xi },
      \quad \text{for} \quad \iota\leq-1, \\
      \psi_a^{\left(0\right)}&=\frac{1}{a-b}+\frac{1}{2\pi {\rm i}}
      \int _\mathbb{C}\int _\mathbb{C}\frac{\psi(\mu) R\left( \mu,\xi \right) }{\xi -a}
       {\rm d}\mu \wedge {\rm d}\overline{\mu }\wedge {\rm d}\xi \wedge {\rm d}\overline{\xi },\\
       \psi_a^{\left(\iota\right)}&=-\frac{1}{(b-a)^{\iota+1}}+\frac{1}{2\pi {\rm i}}
       \int _\mathbb{C}\int _\mathbb{C}\frac{\psi(\mu) R\left( \mu,\xi \right) }{(\xi -a)^{\iota+1}}
       {\rm d}\mu \wedge {\rm d}\overline{\mu }\wedge {\rm d}\xi \wedge {\rm d}\overline{\xi },
       \  \ \text{for}   \ \ \iota \geq 1.
    \end{align}
\end{subequations}
Similar to the treatment in Sec.\ref{sec-2-4}, we have 
\begin{subequations}
\begin{align}
        \psi ^{(\iota)}&=-\bt^T \bK'^{-(\iota+1)} \bB \left(\bI_N+\bv_{N}\bM\bu^{-1}_{N'}
        \bB \right)^{-1}(\bK-b\bI_N)^{-1}\bs\nn,\quad \text{for} \quad \iota\leq-1,\\
         \psi_a^{(0)}&=\frac{1}{a-b}+\bt^T (\bK'-a\bI_{N'})^{-1} \bB \left(\bI_N+\bv_{N}\bM\bu^{-1}_{N'}
         \bB \right)^{-1}(\bK-b\bI_N)^{-1}\bs\nn, \label{phi-a0}\\
    \psi_a^{(\iota)}&=-\frac{1}{(b-a)^{\iota+1}}+\bt^T (\bK'-a\bI_{N'})^{-(\iota+1)}
    \bB \left(\bI_N+\bv_{N}\bM\bu^{-1}_{N'} \bB \right)^{-1}(\bK-b\bI_N)^{-1}\bs\nn,
    \  \ \text{for}  \ \ \iota\geq1.
    \end{align}
\end{subequations}
In Table \ref{tab:my_table2}, we collect in the $N$-soliton solution formulations of the lmKP equation
and the NQC-KP equation derived in this section.

\begin{table}[H]
\begin{tabular}{|l|l|>{\raggedright\arraybackslash}p{0.2\linewidth}|}
\hline
{Equation}   & {Variable} &   {Solution} \\[0.2ex]
\hline
lmKP \eqref{DMKP_U}&$U$&   $1+\psi^{(-1)}$\\[0.2ex]
\hline
NQC-KP \eqref{Sab-2}&${S}(a,b)$&     $\psi(a)$\\[0.2ex]
\hline
\end{tabular}
\vspace{12pt}
\centering
\caption{ {Solutions} of the lmKP  equation and NQC-KP equation}
\label{tab:my_table2}
\end{table}
\noindent
where
\begin{subequations}
\begin{align}
   &  \psi ^{(-1)}= -\bt^T \bB \left(\bI_N+\bv_{N}\bM\bu^{-1}_{N'} \bB \right)^{-1}(\bK-b\bI_N)^{-1} \bs,  \\
   & \psi(a)= \frac{1}{a-b}+\bt^T (\bK'-a\bI_N)^{-1} \bB \left(\bI_N+\bv_{N}\bM\bu^{-1}_{N'} \bB \right)^{-1}
    (\bK-b\bI_N)^{-1}\bs. \label{psi-a}
\end{align}
\end{subequations}

\subsection{Comparison of different formulations}\label{sec-3-4}

In this subsection we briefly compare the solutions from different formulations.
Let us recall Chapter 9.7 in \cite{HJN-book-2016}, where the lKP-type equations are derived
from the Cauchy matrix approach.
One can see that the variable $U_{i,j}$ defined in equation (9.91) in \cite{HJN-book-2016}
corresponds to
\[ \bt^T {\bK'}^{j}\bB \left(\bI_N+\bv_{N}\bM\bu^{-1}_{N'} \bB \right)^{-1}\bK^i\bs\]
in our paper, $S(a,b)$ defined in equation (9.107) in \cite{HJN-book-2016}
corresponds to $\psi(a)$ (see \eqref{psi-a}) in our paper.
Thus, our formulation for the lpKP equation (see $u$ in \eqref{u-phi}) and
our formulation for the NQC-KP equation \eqref{Sab-2} (see ${S}(a,b)$ in \eqref{SS-ab})
are essentially the same as those in \cite{HJN-book-2016}, respectively.

For the lmKP equation given in \eqref{dMKP} and in \eqref{DMKP_U},
the variables $w$ and $U$ correspond to the variables $V(a)$ and $Y(a)$ (see (9.102) in \cite{HJN-book-2016}).
So, they have same formulations.

For the lSKP equation \eqref{dSKP} where $z_c$ is defined in \eqref{z_c-phi}, it is different from
the known formulation (cf. (9.99) in \cite{HJN-book-2016}).

For the NQC-KP equation, we have given three different formulations.
The first one is \eqref{Zab} where $Z(a,b)$ is formulated in \eqref{Z-ab};
the second one is \eqref{NQC-2} where $\mathcal{S}(a,b)$ is formulated in \eqref{S-ab};
the third one is \eqref{Sab-2} where  ${S}(a,b)$ is formulated in \eqref{SS-ab}.
$Z(a,b)$ and  $\mathcal{S}(a,b)$ are apparently different by comparing the definitions of $z_a \theta^{-1}(a,b)$ and
$X(a,b)$.
In addition, from the expressions of $\phi(a), \phi(b)$ and $\psi(a)$
one can also easily see that $\mathcal{S}(a,b)$ and $S(a,b)$ are different.
Thus, equations  \eqref{Zab} and \eqref{NQC-2} are two new formulations for the NQC-KP equation,
or in other words, $Z(a,b)$ and  $\mathcal{S}(a,b)$ provide new solutions to the NQC-KP equation.

\section{Conclusion}\label{sec-4}

In this paper we have explored the eigenfunction equations of the lKP type,
by utilizing $\bar{\partial}$ problems.
We started from the homogeneous $\bar{\partial}$ problem \eqref{DP-1},
which we showed  admits a unique solution with the help of the Cauchy-Pompeiu integral formula.
It follows that this $\bar{\partial}$ problem defines the eigenfunction of the Lax pair
for the lKP equation as well as the lpKP equation.
We then proved that the eigenfunction formulates the lmKP equation.
In addition, the variable $\frac{\phi_c^{(1)}}{\phi_c^{(0)}}$ leads to the lSKP equation and
the NQC-KP equation.
We also showed that $\frac{\phi(b)}{\phi(a)}$ generates a second formulation for the NQC-KP equation.
Besides \eqref{DP-1}, we also considered an inhomogeneous $\bar{\partial}$ problem \eqref{DP-2},
which defines the eigenfunction of the Lax pair for the lmKP equation.
It turns out that such an eigenfunction directly {provides} a third formulation of the NQC-KP equation.

This paper, together with our previous papers \cite{SZZ-MPAG-2025} and \cite{SZZ-PA-2026}, compose a ``trilogy'' of studying discrete eigenfunction equations by utilizing $\bar{\partial}$ problems.
One can also see a hierarchy of eigenfunction equations.
Specifically, the eigenfunction of the lKP (or lpKP) equation defines the lmKP equation,
while the eigenfunction of the lmKP  equation defines the NQC-KP equation.
Such a hierarchy also exists in the 2-dimensional case \cite{SZZ-MPAG-2025,SZZ-PA-2026}.
In addition, the correspondences between the Cauchy matrix
variable $U_{i,j}$ (see equation (9.91) in \cite{HJN-book-2016})
and $\phi(p), \psi(p)$ and the coefficients in their expansions
were also discussed.

In this paper and \cite{SZZ-MPAG-2025,SZZ-PA-2026}, the formulations of solutions are restricted
to solitons, which correspond to discrete exponential functions as plane wave factors (e.g. \eqref{rho}).
There are also elliptic solitons that are generated by the Lam\'e functions
\cite{LSZ-N-2025,NA-IMRN-2010,NSZ-CMP-2023}.
The Dbar formulation for elliptic solitons will be the next topic for investigation.

\vskip 20pt

\subsection*{Acknowledgments}
This project is supported by the NSFC grant (Nos. 1241540016, 12271334 and 12171306)
and the La Trobe University China studies seed-funding research grant.


\vskip 20pt

\begin{thebibliography}{99}

\bibitem{ABF-SAPM-1983} M.J. Ablowitz, D. Bar Yaacov, A.S. Fokas,
       On the inverse scattering transform for the Kadomtsev-Petviashvili equation,
       Stud. Appl. Math., 69 (1983) 135-143.

\bibitem{AF-book-2021} M.J. Ablowitz, A.S. Fokas,
         Introduction to Complex Variables and Applications,
         Camb. Univ. Press, Cambridge, 2021.

\bibitem{ABS-CMP-2003} V.E. Adler, A.I. Bobenko, Yu.B. Suris,
        Classification of integrable equations on quad-graphs. The consistency approach,
        Commun. Math. Phys.,   233  (2003) 513-543.

\bibitem{BC-1980-81} R. Beals, R.R. Coifman,
       Scattering, transformations spectrales et \'equations d$^{\prime}$\'evolution non lin\'eaires,
       S\'eminaire Goulaouic-Schwartz (1980-1981), talk no.22 (9pp).

\bibitem{BC-1981-82} R. Beals, R.R. Coifman,
       Scattering, transformations spectrales et \'equations d$^{\prime}$\'evolution non lin\'eaire  II,
       S\'eminaire Goulaouic-Schwartz (1981-1982), talk no.21 (8pp).

\bibitem{BC-CAPM-1984} R. Beals, R.R. Coifman,
        Scattering and inverse scattering for first order systems,
        Commun. Pure  Appl. Math., 37 (1984) 39-90.

\bibitem{BC-PD-1986} R. Beals, R.R. Coifman,
        The D-bar approach to inverse scattering and nonlinear evolutions,
        Physica D, 18 (1986) 242-249.

\bibitem{BC-IP-1989} R. Beals, R.R. Coifman,
        Linear spectral problems, non-linear equations and the $\bar{\partial}$-method,
        Inverse Probl., 5 (1989) 87-130.

\bibitem{FA-PRL-1981} A.S. Fokas, M.J. Ablowitz,
          Linearization of the Korteweg-de Vries and Painlev{\'e} II equations,
          Phys. Rev. Lett,   47  (1981) 1096-1100.

\bibitem{FA-JMP-1984} A.S. Fokas, M.J. Ablowitz,
        On the inverse scattering transform of multidimensional nonlinear equations related to
        first-order systems in the plane,
        J. Math. Phys., 25 (1984) 2494-2505.

\bibitem{F-PhD-2018} W. Fu,
        Direct Linearisation of Discrete and Continuous Integrable Systems: The KP Hierarchy and its Reductions,
        PhD Thesis, the University of Leeds, 2018.

\bibitem{H-2012} G. Henkin,
        Cauchy-Pompeiu type formulas for $\bar{\partial}$ on affine algebraic Riemann surfaces and some applications,
        in:  Perspectives in Analysis, Geometry and Topology,
        Eds. I. Itenberg, B. J\"oricke,  M. Passare,
        Birkh\"auser, New York, 2012, pp. 213-236.

\bibitem{HJN-book-2016} J. Hietarinta, N. Joshi, F.W. Nijhoff,
     Discrete Systems and Integrability,
     Cambridge University Press, Cambridge, 2016.

\bibitem{Hir-1974}	
         R. Hirota,
         A new form of B\"acklund transformations and its relation to the inverse scattering problem,
         Prog. Theo. Phys., 52 (1974) 1498-1512.

\bibitem{JM-IP-1986} M. Jaulent, M. Manna,
      Connection between the KdV and the AKNS spatial transforms,
      Inverse Probl., 2 (1986) L35-L41.

\bibitem{JM-EPL-1986} M. Jaulent, M. Manna,
      The spatial transform method for multidimensional $(2+1)$ problems,
      Europhys. Lett., 2 (1986) 891-896.

\bibitem{JM-IP-1987} M. Jaulent, M. Manna,
      The spatial transform method for $\bar{\partial}$ equations of $n$th linear order,
      Inverse Probl., 3 (1987) L13-L18.

\bibitem{JM-JMP-1987} M. Jaulent, M. Manna,
      The ``spectral Wronskian'' tool and the $\bar{\partial}$ investigation of the KdV hierarchy,
      J. Math. Phys., 28 (1987) 2338-2342.

\bibitem{JMM-IP-1988} M. Jaulent, M. Manna, L. Martinez Alonso,
       $\bar{\partial}$ equations in the theory of integrable systems,
       Inverse Probl., 4 (1988) 123-150.

\bibitem{K-RMP-1990} B.G. Konopelchenko,
         Soliton eigenfunction equations: the IST integrability and some properties,
         Rev. Math. Phys., 2 (1990) 399-440.

\bibitem{LSZ-N-2025}	X. Li, Y.Y. Sun, D.J. Zhang,
       The direct linearization scheme with Lam\'e function: The KP equation and reductions,
       Nonlinearity 38  (2025) 105024 (30pp).

\bibitem{MiwJD-2000} 
         T. Miwa, M. Jimbo, E. Date,
         Solitons: Differential Equations, Symmetries and Infinite Dimensional Algebras,
         Camb. Univ. Press, Cambridge, 2000.

\bibitem{NA-IMRN-2010} F.W. Nijhoff, J. Atkinson,
       Elliptic $N$-soliton solutions of ABS lattice equations,
       Int. Math. Res. Not., 2010 (2010) 3837-3895.


\bibitem{NCWQ-PLA-1984} F.W. Nijhoff, H. Capel, G. Wiersma, R. Quispel,
        B\"acklund transformations and three-dimensional lattice equations,
         Phys. Lett. A, 105 (1984) 267-272.


\bibitem{NSZ-CMP-2023} F.W. Nijhoff, Y.Y. Sun, D.J. Zhang,
         Elliptic solutions of Boussinesq type lattice equations and the elliptic $N$th  root of unity,
         Commun. Math. Phys., 399  (2023) 599-650.

\bibitem{N-JPA-2006} J.J.C. Nimmo,
         On a non-Abelian Hirota-Miwa equation,
         J. Phys. A: Math. Gen., 39 (2006), 5053-5065.


\bibitem{P-1909} D. Pompeiu,
      Sur la repr\'esentation des fonctions analytiques par des int\'egrales d\'efinies,
      C. R. Acad. Sc. Paris, 149 (1909) 1355-1357.

\bibitem{SZZ-MPAG-2025} L.L. Shi, C. Zhang, D.J. Zhang,
      On the $\bar{\partial}$ method and direct linearization approach of the lattice KdV type equations,
      Math. Phys. Anal. Geom., 28 (2025) 40 (32pp).

\bibitem{SZZ-PA-2026} L.L. Shi, C. Zhang, D.J. Zhang,
      On the formulation of the NQC variable,
      Physica A: Stat. Mech. Appl., 685 (2026) 131317 (15pp).

\bibitem{WZZZ-JPA-2024} X.Y. Wu, C. Zhang, D.J. Zhang, H.F. Zhang,
        Lattice eigenfunction equations of KdV-type,
        J. Phys. A: Math. Theor.,  57  (2024) 255202 (25pp).

\bibitem{ZM-FAA-1985} V.E. Zakharov, S.V. Manakov,
        Construction of higher-dimensional nonlinear integrable systems and of their solutions,
        Funct. Anal. Appl.,  19  (1985) 89-101.



\end{thebibliography}
\end{document}